\documentclass[12pt,twoside,english,reqno]{article}
\usepackage[T1]{fontenc}
\usepackage[latin9]{inputenc}
\usepackage{amsthm}
\usepackage{amsmath}
\usepackage{amssymb}
\usepackage{esint}

\makeatletter

\DeclareRobustCommand{\cyrtext}{%
  \fontencoding{T2A}\selectfont\def\encodingdefault{T2A}}
\DeclareRobustCommand{\textcyr}[1]{\leavevmode{\cyrtext #1}}
\AtBeginDocument{\DeclareFontEncoding{T2A}{}{}}

\numberwithin{equation}{section}
\numberwithin{figure}{section}
  \theoremstyle{remark}
  \newtheorem*{rem*}{Remark}
\theoremstyle{plain}
\newtheorem{thm}{Theorem}
  \theoremstyle{plain}
  \newtheorem{lem}[thm]{Lemma}
  \theoremstyle{plain}
  \newtheorem{prop}[thm]{Proposition}
  \theoremstyle{plain}
  \newtheorem{cor}[thm]{Corollary}
\newcommand{\lyxaddress}[1]{
\par {\raggedright #1
\vspace{1.4em}
\noindent\par}
}

\usepackage{amscd}\usepackage{amsbsy}\usepackage{amsthm}

\usepackage{amsfonts}

\newtheorem{theorem}{Theorem}

\makeatother

\usepackage{babel}

\begin{document}

\title{A class of nonlinear random walks related to the Ornstein-Uhlenbeck
process }

\author{S.A. Muzychka, K.L. Vaninsky}

\maketitle

\begin{abstract}
Contrary to the theory of Markov processes, no general theory exists
for the so called nonlinear Markov processes. We study an example
of {}``nonlinear Markov process'' related to classical probability
theory, merely to random walks. This model provides interesting phenomena
(absent in classical Markov chains): continuum of stationary measures,
conserved quantities, convergence to stationary classical random walks
etc.
\end{abstract}
\tableofcontents{}\setcounter{section}{0} \setcounter{equation}{0}

\section{Introduction}

Contrary to the theory of Markov processes, no general theory exists
for the so called nonlinear Markov processes. Though a general definition
of a nonlinear Markov process was introduced by H.P. McKean, {[}4{]}
in his study of various models of kinetic theory. Subsequently various
authors, see {[}5{]},{[}6{]}, considered limits of stochastic many
particles systems which lead to processes of this special type.

Here we give an example of {}``nonlinear Markov process'' which
is close to classical probability theory, merely to random walks.
It appears as a mathematical model of a market with two type of agents
or participants, traditionally called bulls and bears. This model
provides interesting phenomena (absent in classical Markov chains):
continuum of stationary measures, conserved quantities, convergence
to stationary classical random walks etc. It is important that our
system has some relation to the Ornstein-Uhlenbeck process. This underlies
main intuition and makes the system solvable.

\subsection{Simple random walks on $\mathbb{Z}$ with {}``discrete gaussian''
stationary measure }

Consider a continuous time Markov chain $\eta_{t}$ (simple random
walk) on $\mathbb{Z}$. The intensity of the jumps $n\to n+1,$ $n\to n-1$
are correspondingly \[
\lambda_{n}=e^{-c(n-L)};\quad\mu_{n}=e^{c(n-M)},\]
 where $c>0$ and $L$ and $M$ are real numbers. The chain is ergodic
and reversible. The detailed balance equations \begin{equation}
\pi(n)\lambda_{n}=\pi(n+1)\mu_{n+1}\label{detBalance}\end{equation}
for the stationary measure $\pi$ have the unique solution \begin{equation}
\pi(n)=\frac{1}{\Xi}e^{-c(n-s)^{2}},\qquad\qquad s=\frac{L+M}{2}.\label{eq:invariantMeasure}\end{equation}
The normalization factor $\Xi=\Xi(s,c)$ is given by \[
\Xi(s,c)=e^{-cs^{2}}\Theta\left(\frac{cs}{i\pi},\frac{ci}{\pi}\right),\]
 where \[
\Theta(v,\tau)=\sum e^{2\pi ivn+\pi i\tau n^{2}}\]
is the Jacobi theta function, see {[}1{]} p.188.

In addition to $s$ we introduce another variable $d=\frac{L-M}{2}$.
It is interesting that the invariant measure $\pi,$ which should
depend on both parameters $L$ and $M,$ depends on $s$ only.

Let us note that the invariant measure does not change under the following
transformation of the jump rates\[
\lambda_{n}\to\lambda_{n}\beta(n),\mu_{n}\to\mu_{n}\beta(n-1),\]
where $\beta(n)$ is an arbitrary positive function. This follows
from the detailed balance equations for the invariant measure. In
particular, $\beta(n)$ can be chosen so that the mean drift becomes
asymptotically linear\[
m(n)=\lambda_{n}-\mu_{n}\sim-Cn,\quad C>0\]
as for the classical Ornstein-Uhlenbeck process. Remind that the Ornstein-Uhlenbeck
process is the unique stationary gaussian Markov process on $\mathbb{R}$.

\subsection{Nonlinear walks and main results.}

Consider the vector-function\[
X(t)=(L(t),M(t),p_{n}(t),n\in\mathbb{Z})\]
 with $(2+\infty)$ real functions on the interval $t\in[0,\infty)$
and denote\begin{equation}
\lambda_{n}(t)=\beta(n)e^{-c(n-L(t))},\mu_{n}(t)=\beta(n-1)e^{c(n-M(t))},\label{eq:lambdaMu}\end{equation}
where $c>0$ is some constant.

The vector-function $X(t)$ is defined by the following infinite system
of ordinary differential equations\begin{equation}
\frac{dp_{n}}{dt}=\lambda_{n-1}p_{n-1}-(\lambda_{n}+\mu_{n})p_{n}+\mu_{n+1}p_{n+1},n\in\mathbb{Z}\label{mainp}\end{equation}
\begin{equation}
\frac{dL}{dt}=-\sum_{n\in\mathbb{Z}}p_{n}\lambda_{n}+C_{\lambda}\label{mainL}\end{equation}
\begin{equation}
\frac{dM}{dt}=\sum_{n\in\mathbb{Z}}p_{n}\mu_{n}-C_{\mu}\label{mainM}\end{equation}
 together with the initial conditions $L(0),M(0),p_{n}(0)$. We will
assume that \[
p_{n}(0)\geq0,\sum_{n\in\mathbb{Z}}p_{n}(0)=1.\]
Otherwise speaking, $p_{n}(0)$ define the probability measure $p(0)$
on $\mathbb{Z}$.

Apriori, $C_{\lambda}$ and $C_{\mu}$ are some positive constants.
If however, there exists at least one fixed point $(L,M,\pi)$ for
these equations, then\[
C_{\lambda}=\sum_{n\in\mathbb{Z}}\pi_{n}\lambda_{n},C_{\mu}=\sum_{n\in\mathbb{Z}}\pi_{n}\mu_{n}.\]
Then $\pi_{n}$ satisfy equations (for fixed $L,M$)\[
\lambda_{n-1}\pi_{n-1}-(\lambda_{n}+\mu_{n})\pi_{n}+\mu_{n+1}\pi_{n+1}=0,\]
which look exactly as Kolmogorov equations for stationary probabilities
of the countable Markov chain. It is known (see {[}3{]}, p. 59, th.
7.1) that the only $l_{1}$-solution of these equations is positive
(up to some multiplicative constant). Thus $\pi_{n}$ satisfy also
the detailed balance equations (\ref{detBalance}). It follows that
$C_{\lambda},C_{\mu}$ are equal\[
C_{\lambda}=\sum_{n\in\mathbb{Z}}\pi_{n}\lambda_{n}=\sum_{n\in\mathbb{Z}}\pi_{n}\mu_{n}=C_{\mu}.\]
Then using the variables $s,d$ introduced above we rewrite $(\ref{mainp}-\ref{mainM})$
in the following form:

\begin{equation}
\begin{cases}
p_{n}'(t)= & e^{cd}[\beta(n-1)e^{c(-n+1+s)}p_{n-1}-(\beta(n)e^{c(-n+s)}+\beta(n-1)e^{c(n-s)})p_{n}+\\
 & +\beta(n)e^{c(-n-1+s)}p_{n+1}],\quad n\in\mathbb{Z};\\
s'(t)= & -\frac{1}{2}e^{cd}\left(\sum_{n\in\mathbb{Z}}p_{n}\beta(n)e^{c(-n+s)}-\sum_{n\in\mathbb{Z}}p_{n}\beta(n-1)e^{c(n-s)}\right);\\
d'(t)= & -\frac{1}{2}e^{cd}\left(\sum_{n\in\mathbb{Z}}p_{n}\beta(n)e^{c(-n+s)}+\sum_{n\in\mathbb{Z}}p_{n}\beta(n-1)e^{c(n-s)}\right)+C_{\lambda},\end{cases}\label{eq:mainEquationsinTermsSandD}\end{equation}
where we have assumed the existence of the fixed point. First two
equations show that the trajectory of a pair $(p,s)$ does not depend
on $d$. This observation will help us in the proof of convergence.
\begin{rem*}
Now we want to explain some market model, which is the source of this
paper. Assume that on the integer lattice $\mathbb{Z}$ all points
of the interval $(-\infty,b]$ are occupied by {}``bulls'' who want
to buy and the points on the interval $[b+1,\infty)$ who want to
sell. The boundary $b=b(t)$ changes with time as follows. There are
two Poisson arrival streams of demands: to buy with the rate $\lambda_{b}$
and to sell with the rate $\mu_{b}$. When the buy demand arrives
the boundary immediately moves $b\to b+1$, and conversely. The parameters
$L$ and $M$ reflect the opinion of bulls and bears correspondingly,
concerning the fair price.
\end{rem*}
Define the Banach space $\mathbf{B}$ of vector-functions $p=\left\{ p_{n}\right\} _{n\in\mathbb{Z}}$
with the norm\[
\|p\|_{\alpha}=\sum_{n\in\mathbb{Z}}|p_{n}|\exp(\frac{n^{2}}{2}+\alpha|n|),\quad\alpha\in\mathbb{R}.\]
Throughout this paper we assume that $\beta(n)$ satisfies the following
condition\begin{equation}
\sup_{n\in\mathbb{Z}}\beta(n)<\infty.\label{con1}\end{equation}

\begin{theorem}

For any initial conditions such that $p(0)\in\mathbf{B}$ is the probability
measure, the solution of the system (\ref{eq:mainEquationsinTermsSandD})
exists on the interval $[0,\infty)$ and is unique in the space $\mathbf{B}\times C^{2}([0,\infty))=\left\{ (p,L,M)\right\} .$

Moreover, for any $t$ the quantities $p_{n}(t)$ define the probability
measure $p(t),$ that is $p_{n}(t)\geq0,\sum_{n}p_{n}(t)=1.$

\end{theorem}

\begin{theorem}

If $C_{\lambda}\neq C_{\mu}$ there are no fixed points. If $C_{\lambda}=C_{\mu}>0$
the set of fixed points is a one parameter family $\{(L_{s},M_{s},\pi_{s}(n))\}$,
which depends on the parameter $s\in\mathbb{R}$. It is given explicitely
by \[
\pi_{s}(n)=\frac{1}{\Xi}e^{-c(n-s)^{2}},\qquad\Xi=\sum_{n\in\mathbb{Z}}e^{-c(n-s)^{2}};\]
\[
L_{s}=s+\ln\left[C_{\lambda}\left(\frac{\sum_{l}e^{-c(l-s)^{2}}}{\sum_{k}\beta(k)e^{-c(k-s)^{2}}e^{c(-k+s)}}\right)\right];\]

\[
M_{s}=s-\ln\left[C_{\lambda}\left(\frac{\sum_{l}e^{-c(l-s)^{2}}}{\sum_{k}\beta(k)e^{-c(k-s)^{2}}e^{c(-k+s)}}\right)\right].\]
Moreover, $s=\frac{L_{s}+M_{s}}{2}$.

\end{theorem}

\begin{theorem}

If $C_{\lambda}=C_{\mu}>0$ then there is a conserved quantity (invariant
of motion)\[
K=K(X)=L+M+\sum_{n\in\mathbb{Z}}np_{n}.\]
Any hypersurface defined by the value of $K(X)$ contains exactly
one fixed point.

\end{theorem}

Speaking otherwise, the conserved quantity makes our phase space a
fiber bundle over the real line, where each fiber contains exactly
one fixed point.

For the next theorem we need, besides condition (\ref{con1}), the
following condition: there is a positive constant $C>0$ such that
for all $n\in\mathbb{Z}$ \begin{equation}
\inf_{n\in\mathbb{Z}}\beta(n)>0,\quad\frac{1}{e}\beta(n+1)-\beta(n)<-C,\quad\frac{1}{e}\beta(n-1)-\beta(n)<-C.\label{con2}\end{equation}

This (very technical) assumption we will need only for proving convergence.
Note that unfortunately this conjecture does not cover the case of
linear drift, but $\beta(n)\equiv1$ satisfies ($\ref{con2}$).

\begin{theorem}

Assume condition $(\ref{con2})$. Then for any initial point $X(0)$
such that the initial probability measure $p(0)\in\mathbf{B}$ the
solution converges to the unique fixed point on the hypersurface defined
by the value of $K(X(0))$.

\end{theorem}

\begin{theorem}

For any initial conditions $X(0)$ such that the initial probability
measure $p(0)\in\mathbf{B}$ there exists a random process $\xi(t)=\xi(t,X(0))\in\mathbb{Z},t\in[0,\infty),$
with probability meausure $P=P_{X(0)}$ on the set $X(t)$ of trajectories
such that \[
P(\xi(t)=n)=p_{n}(t).\]
A such that the $k$-dimensional distributions of $\xi(t)$, for $k>1$,
are defined in Markovian way by\begin{equation}
P_{X(0)}(\xi(t_{1})=n_{1},...,\xi(t_{k})=n_{k})=p_{n_{1}}(t_{1})P_{X(0)}(n_{2},t_{2}|n_{1},t_{1})...P_{X(0)}(n_{k},t_{k}|n_{k-1},t_{k-1}).\label{k-dim-distr}\end{equation}
Under condition (\ref{con2}), the $k$-dimensional distributions
of $\xi(t)$ tend as $t\to\infty$ to the corresponding $k$-dimensional
distributions of the stationary Markov process $\eta_{t}$ defined
above.

\end{theorem}

Let us note that while proving Theorem 1, we construct a family $P_{X(0)}(n,s|m,t),$
$t<s,m,n\in\mathbb{Z},$ of stochastic matrices satisfying the semigroup
property. Thus the latter theorem is just the definition of the process
$\xi(t)$, Formula (\ref{k-dim-distr}) looks like it defines a time
inhomogeneous Markov process, but in fact it does not, since the transition
kernels $P_{X(0)}(\cdot,\cdot|\cdot,\cdot)$ depend on the initial
conditions.

\textbf{Acknowledgments.} KV would like to thank Yuri Suhov, Henry
McKean and Raghy Varadhan for stimulating discussions. Both authors
would like to thank Vadim Malyshev for his interest in this work.

\section{Proofs}

Everywhere we will omit the parameter $c$ assuming $c=1.$ To simplify
notation we denote the pair of functions $L$ and $M$ by $Z(t)=(L(t),M(t))$.

\subsection{Existence and uniqueness}

Here we will prove Theorem 1. The scheme of the proof is the following.
Assuming that the continuous functions $L(t),M(t)$ are given, we
prove that the solution of (\ref{mainp}) exists and is unique in
the appropriate Banach space, moreover it has some necessary properties
in this space. Then we substitute this solution to the equations (\ref{mainL}-\ref{mainM}),
thus obtaining two ODE with two unknown functions, and prove that
the solution of these two equations exists.

\paragraph{Two Banach spaces.}

Consider the Banach space $B_{\alpha}^{+},$ which consists of infinite
sequences $(\nu_{k},k\in\mathbb{Z})$ of real numbers with the norm:
\[
\|\nu\|_{\alpha}^{+}=\sum_{k\in\mathbb{Z}}e^{\frac{k^{2}}{2}+\alpha|k|}|\nu_{k}|,\]
 and the Banach space $B_{\alpha}^{-}$ with the norm: \[
\|f\|_{\alpha}^{-}=\sum_{k\in\mathbb{Z}}e^{-\frac{k^{2}}{2}-\alpha|k|}|f_{k}|.\]
 Everywhere below $\alpha$ is an arbitrary fixed real number. Let
us explain the meaning of these Banach spaces. $B_{\alpha}^{+}$ is
the space of admissible measures of the process. $B_{\alpha}^{-}$
is the space of admissible functions. The natural duality between
$B_{\alpha}^{+}$ and $B_{\alpha}^{-}$ is \[
\langle\nu,f\rangle=\sum_{n\in\mathbb{Z}}\nu_{n}f_{n},\qquad\qquad\nu\in B_{\alpha}^{+},f\in B_{\alpha}^{-}.\]
 It is easy to see that \[
|\langle\nu,f\rangle|\leq\|\nu\|_{\alpha}^{+}\|f\|_{\alpha}^{-}.\]
The space of bounded operators, acting on $B_{\alpha}^{+}$ and $B_{\alpha}^{-},$
we denote by $\mathcal{L}(B_{\alpha}^{+})$ and $\mathcal{L}(B_{\alpha}^{-})$
correspondingly. The operators are acting on $B_{\alpha}^{+}$ from
the right, and on $B_{\alpha}^{-}$ from the left.

Finally note that for any $\alpha_{1},\alpha_{2}\in\mathbb{R},$ such
that $\alpha_{1}>\alpha_{2}$, the following inclusions hold $B_{\alpha_{1}}^{+}\subset B_{\alpha_{2}}^{+}$
and $B_{\alpha_{2}}^{-}\subset B_{\alpha_{1}}^{-}.$ We will use these
properties below.

\paragraph{Transition probabilities.}

Assume now that $L(t)$ and $M(t)$ are some fixed continuous functions
on $\mathbb{R}_{+}$. We will prove that the Markov process, defined
by the Kolmogorov equations (\ref{mainp}) for $p_{n}(t)$ exists
and is unique in $B_{\alpha}^{+}$. Denote by $P(t,s),$ $t\leq s$
the family of its transition probability matrixes, $H(t)$ - the infitesimal
matrix. Let $H_{0}(t)$ and $V(t)$ be a diagonal and off diagonal
parts of $H(t)$ correspondingly.

First we present a useful formula for the transition probabilities
valid for a denumerable inhomogeneous continuous time Markov chain.
Denote $\triangle_{k}(t,s)=\{(s_{1},\ldots,s_{k})\in\mathbb{R}^{k}:$
$t\leq s_{k}\leq\ldots\leq s_{1}\leq s\}$ the $k$-dimensional simplex.
\begin{lem}
\label{lem:usefulFormula}Let $X_{t}$ be a continuous time inhomogeneous
Markov chain with denumerable state space and the family of transition
probability matrices $P(t,s),$ defined for $0\leq t\leq s<\infty.$
Denote by $H(t)$ an infinitesimal matrix of $X_{t}$. Let $H_{0}(t)$
and $V(t)$ be a diagonal and off diagonal parts of $H(t),$ then
for any $t<s$ the series \[
P(t,s)=e^{\int_{t}^{s}H_{0}(s)ds}+\]
\begin{equation}
+\sum_{k=1}^{\infty}\int_{\triangle_{k}(t,s)}e^{\int_{t}^{s_{k}}H_{0}(s)ds}V(s_{k})e^{\int_{s_{k}}^{s_{k-1}}H_{0}(s)ds}\ldots V(s_{1})e^{\int_{s_{1}}^{s}H_{0}(s)ds}ds_{k}\ldots ds_{1}\label{trivial_series}\end{equation}
is absolutely norm convergent for some norm $\|\cdot\|$, if $\sup_{u\in[t,s]}\|V(u)\|<\infty$. \end{lem}
\begin{proof}
Formally the series is obtained by the iteration of the following
formula\[
P(t,s)-e^{\int_{t}^{s}H_{0}(s)ds}=\int_{t}^{s}P(t,z)V(z)e^{\int_{z}^{s}H_{0}(s)ds}dz.\]
Since all diagonal terms of $e^{\int H_{0}(s)ds}$ do not exceed 1,
then using $\sup_{u\in[t,s]}\|V(u)\|<\infty,$ and the formula for
the volume of the simplex we get the result.
\end{proof}
We want to prove that the corresponding series converges in $B_{\alpha}^{+}$
and the matrices $P(t,s)$ are stochastic and satisfy the Chapman-Kolmogorov
equations. In order to do this we have to check that $V(t)$ are bounded
operators in $B_{\alpha}^{+}$. First we will prove a technical lemma
which will explain the condition (\ref{con1}).

Consider the following infinite three-diagonal matrix

\[
V=\left(\begin{array}{ccccc}
... & ... & ... & ... & ...\\
\mu_{n-1} & 0 & \lambda_{n-1}\\
 & \mu_{n} & 0 & \lambda_{n}\\
 &  & \mu_{n+1} & 0 & \lambda_{n+1}\\
... & ... & ... & ... & ...\end{array}\right)\]

We will consider $V$ as the operator acting on infinite sequences
from the right and from the left.
\begin{lem}
\label{lem:new}There exists a sequence $\left\{ c_{n}\right\} ,n\in\mathbb{Z}$
$c_{n}>0$ such that $V$ is a bounded operator in the Banach space
with the norm\[
\|x\|=\sum_{n\in\mathbb{Z}}c_{n}|x_{n}|\]

if and only if $\sup_{n\in\mathbb{Z}}\lambda_{n}\mu_{n+1}<\infty.$ \end{lem}
\begin{proof}
We will prove this lemma for the case when $V$ is acting from the
right. For the action from the left the proof is similar.

\textbf{Necessity.} Let $V$ be bounded. Then if $e_{n}=\delta_{0,n}$

\[
\frac{\|e_{n}V\|}{\|e_{n}\|}=\frac{\|\mu_{n}e_{n-1}+\lambda_{n}e_{n+1}\|}{\|e_{n}\|}=\frac{c_{n-1}}{c_{n}}\mu_{n}+\frac{c_{n+1}}{c_{n}}\lambda_{n}\leq\|V\|=\mathrm{const.}\]

Whence we have a double inequality \[
\frac{\mu_{n+1}}{\|V\|}\leq\frac{c_{n+1}}{c_{n}}\leq\frac{\|V\|}{\lambda_{n}}.\]

This gives the necessary conclusion.

\textbf{Sufficiency.} Assume that $\sup_{n\in\mathbb{Z}}\lambda_{n}\mu_{n+1}<\infty.$
A straightforward calculation shows that for \[
c_{n}=\sqrt{\frac{\mu_{1}\ldots\mu_{n}}{\lambda_{0}\ldots\lambda_{n-1}}},\]

we get \[
\frac{\|e_{n}V\|}{\|e_{n}\|}=\sqrt{\lambda_{n-1}\mu_{n}}+\sqrt{\lambda_{n}\mu_{n+1}}<\infty.\]

\end{proof}
Applying this lemma to the our case we see that the condition (\ref{con1})
is spelling natural. Indeed \[
\lambda_{n}(t)\mu_{n+1}(t)=\beta^{2}(n)e^{L(t)-M(t)},\]
therefore the condition of Lemma $\ref{lem:new}$ is equivalent to
(\ref{con1}).
\begin{lem}
\label{lemma2} Consider the operator valued function $V(t),$ defined
above. Then this function takes values in the set of bounded operators
in $\mathcal{L}(B_{\alpha}^{+})$, or in $\mathcal{L}(B_{\alpha}^{-})$.
Moreover, it is continuous and \[
\|V(t)\|_{\alpha}^{\pm}\leq\mathrm{const}(e^{-M(t)}+e^{L(t)})\]
 for any $t$. \end{lem}
\begin{proof}
Consider an arbitrary vector $\nu\in B_{\alpha}^{+}$. We have \[
\begin{split} & \|\nu V(t)\|_{\alpha}^{+}=\sum_{k\in\mathbb{Z}}e^{\frac{k^{2}}{2}+\alpha|k|}|(\nu V(t))_{k}|=\\
 & =\sum_{k\in\mathbb{Z}}e^{\frac{k^{2}}{2}+\alpha|k|}|e^{-k+1+L(t)}\beta(k-1)\nu_{k-1}+\mbox{}e^{k+1-M(t)}\beta(k)\nu_{k+1}|\leq\\
 & \leq\mathrm{const}\sum_{k\in\mathbb{Z}}e^{\frac{k^{2}}{2}+\alpha|k|-k+1+L(t)}|\nu_{k-1}|+\mathrm{const}\sum_{k\in\mathbb{Z}}e^{\frac{k^{2}}{2}+\alpha|k|+k+1-M(t)}|\nu_{k+1}|=\\
 & =\mathrm{const}\sum_{k\in\mathbb{Z}}(e^{\frac{(k+1)^{2}}{2}+\alpha|k+1|-k-1+1+L(t)}+e^{\frac{(k-1)^{2}}{2}+\alpha|k-1|+k-1+1-M(t)})|\nu_{k}|=\\
 & =\mathrm{const}\sum_{k\in\mathbb{Z}}(e^{\frac{k^{2}}{2}+\frac{1}{2}+\alpha|k+1|+L(t)}+e^{\frac{k^{2}}{2}+\frac{1}{2}+\alpha|k-1|-M(t)})|\nu_{k}|\leq\mathrm{const}(e^{L(t)}+e^{-M(t)})\|\nu\|_{\alpha}^{+}.\end{split}
\]
 Similar calculation for any $f\in B_{\alpha}^{-}$ implies the inequality
for $\|V(t)\|_{\alpha}^{-}.$

It remains to check that $V$ is continuous in $t.$ We will do it
for the space $B_{\alpha}^{+}$ only. For $B_{\alpha}^{-}$ it can
be verified along the same lines. As in the estimates above for any
arbitrary nonzero $\nu\in B_{\alpha}^{-},$ and arbitrary $t_{1},t_{2}$
we have: \[
\|\nu V(t_{1})-\nu V(t_{2})\|_{\alpha}^{+}\leq\mathrm{const}(|e^{L(t_{1})}-e^{L(t_{2})}|+|e^{-M(t_{1})}-e^{-M(t_{2})}|)\|\nu\|_{\alpha}^{+}.\]
 Together with the fact that $L(t)$ and $M(t)$ are continuous, this
implies our statement.\end{proof}
\begin{lem}
\label{lemma3} Let $t\leq s.$ Then the series (\ref{trivial_series})
converges in both norms of $\mathcal{L}(B_{\alpha}^{+})$ and $\mathcal{L}(B_{\alpha}^{-})$
and therefore defines the bounded operator. Moreover, \[
\|P(t,s)\|_{\alpha}^{\pm}\leq\exp(\mathrm{const}(s-t)\sup_{t\in[0,s]}(e^{-M(t)}+e^{L(t)})).\]
\end{lem}
\begin{proof}
Let us prove the lemma for $B_{\alpha}^{+}.$ Arguments for $B_{\alpha}^{-}$
are exactly the same. First, note that Lemma \ref{lemma2} implies
that for any $t\leq s$ \begin{equation}
\sup_{t\in[0,s]}\|V(t)\|_{\alpha}^{+}\leq\mathrm{const}\cdot\sup_{t\in[0,s]}(e^{-M(t)}+e^{L(t)}).\label{eq:estimateForV}\end{equation}
 Since $H_{0}$ consists of negative numbers, then for any $t_{1}\leq t_{2}$
we have $\|e^{\int_{t_{1}}^{t_{2}}H_{0}(s)ds}\|_{\alpha}^{+}\leq1.$
Using that the volume of the simplex $\triangle_{k}(t,s)$ is $\frac{(s-t)^{k}}{k!},$
and the estimates above, we obtain: \[
\|P(t,s)\|_{\alpha}^{+}\leq1+\sum_{k=1}^{\infty}\frac{(\sup_{[0,s]}\|V(t)\|_{\alpha}^{+})^{k}(s-t)^{k}}{k!}=e^{(s-t)\sup_{[0,s]}\|V(t)\|_{\alpha}^{+}}.\]

\end{proof}

\paragraph{Approximation by finite Markov chains.}

Lemma \ref{lemma3} states that for any $t\leq s$ the operator $P(t,s)$
is defined in the spaces $B_{\alpha}^{+}$. It remains to prove that
they define a Markov process. For the proof we need to introduce new
notation.

Define truncated Markov processes $X^{m}$ as the restriction of $X$
on $[-m,m]$. More exactly $X^{m}$ has the infinitesimal rates \[
k\rightarrow k+1:\quad\lambda_{k}^{m}(t)=\beta(k)e^{-k+L(t)},\quad k\in[-m,m-1];\]
 \[
k\rightarrow k-1:\quad\mu_{k}^{m}(t)=\beta(k-1)e^{k-M(t)},\quad k\in[-m+1,m].\]
 Let $H^{m}=H^{m}(t)$ be the infinitesimal matrix of $X^{m}$. Similar
to what we have done before we write $H^{m}$ in the form $H^{m}=H_{0}^{m}+V^{m},$
where $H_{0}^{m},V^{m}$ are its diagonal and off diagonal parts.
For $X^{m}$ obviously holds the formula analogous to (\ref{trivial_series})
\begin{multline*}
P^{m}(t,s)=e^{\int_{t}^{s}H_{0}^{m}(s)ds}+\\
+\sum_{k=1}^{\infty}\int_{\triangle_{k}(t,s)}e^{\int_{t}^{s_{k}}H_{0}^{m}(s)ds}V^{m}(s_{k})\ldots V^{m}(s_{1})e^{\int_{s_{1}}^{s}H_{0}^{m}(s)ds}ds_{k}\ldots ds_{1}.\end{multline*}

\begin{lem}
\label{lemma4} \textbf{1.} For any $\pi\in B_{\alpha}^{+}$ and any
$t\leq s$: $\pi P^{m}(t,s)\to_{m\to\infty}\pi P(t,s)$ in the sense
of the norm $\|\cdot\|_{\alpha}^{+}$;

\textbf{2.} For any $\pi\in B_{\alpha}^{-}$ and any $t\leq s$: $P^{m}(t,s)\pi\to_{m\to\infty}P(t,s)\pi$
in the norm $\|\cdot\|_{\alpha}^{-}$.
\end{lem}
\textit{Proof.} We will give a proof only for $B_{\alpha}^{+}.$ For
$B_{\alpha}^{-}$ the proof is the same. Fix some $0\leq t\leq s$
and define\[
\Gamma^{m}(s_{1},\ldots,s_{k}):=\pi e^{\int_{t}^{s_{k}}H_{0}(s)ds}V(s_{k})e^{\int_{s_{k}}^{s_{k-1}}H_{0}(s)ds}\ldots V(s_{1})e^{\int_{s_{1}}^{s}H_{0}(s)ds}-\]
\[
-\pi e^{\int_{t}^{s_{k}}H_{0}^{m}(s)ds}V^{m}(s_{k})e^{\int_{s_{k}}^{s_{k-1}}H_{0}^{m}(s)ds}\ldots V^{m}(s_{1})e^{\int_{s_{1}}^{s}H_{0}^{m}(s)ds},\]
where $\{s_{j}\}_{j=1}^{k}\in\triangle_{k}(t,s).$ It is easy to check
that $\pi e^{\int_{t_{1}}^{t_{2}}H_{0}^{m}(s)ds}\rightarrow_{m\to\infty}\pi e^{\int_{t_{1}}^{t_{2}}H_{0}(s)ds}$
in the sense of the norm $\|\cdot\|_{\alpha}^{+}$ for $t_{1}\leq t_{2}$,
and $\pi V^{m}(t)\rightarrow_{m\to\infty}\pi V(t)$. Therefore in
the norm \[
\Gamma^{m}(s_{1},\ldots,s_{k})\to_{m\to\infty}0\]
 for all sets $\{s_{j}\}_{j=1}^{k},$ which belong to correspondent
simplex. Let us estimate the difference of the $k$-th terms $A_{k}$
and $A_{k}^{m}$ of the series for $\pi P(t,s)$ and $\pi P^{m}(t,s)$
correspondingly \[
\|A_{k}-A_{k}^{m}\|_{\alpha}^{+}\leq\int_{\triangle_{k}(t,s)}\|\Gamma^{m}(s_{1},\ldots,s_{k})\|_{\alpha}^{+}ds_{k}\ldots ds_{1}.\]
 Using the estimates, similar to the one used in Lemma ~\ref{lemma3},
it is easy to check that $\|\Gamma^{m}(s_{1},\ldots,s_{k})\|_{\alpha}^{+}$
bounded on simplex, namely: \[
\|\Gamma^{m}(s_{1},\ldots,s_{k})\|_{\alpha}^{+}\leq2\|\pi\|_{\alpha}^{+}(\mathrm{const\cdot}\sup_{[0,s]}(e^{-M(t)}+e^{L(t)}))^{k}.\]
 Therefore, Lebesgue theorem implies that: \[
\int_{\triangle_{k}(t,s)}\|\Gamma^{m}(s_{1},\ldots,s_{k})\|_{\alpha}^{+}ds_{k}\ldots ds_{1}\to_{m\to\infty}0,\]
 i.e. \[
\|A_{k}-A_{k}^{m}\|_{\alpha}^{+}\to_{m\to\infty}0.\]
 Moreover, using a formula for the volume of simplex it is easy to
get, that: \[
\|A_{k}-A_{k}^{m}\|_{\alpha}^{+}\leq\frac{s^{k}}{k!}2\|\pi\|_{\alpha}^{+}(\mathrm{const\cdot}\sup_{[0,s]}(e^{-M(t)}+e^{L(t)}))^{k},\]
 i.e. $\sum_{k=1}^{\infty}\|A_{k}-A_{k}^{m}\|_{\alpha}^{+}$ converges
uniformly in $m.$

We will use the following simple
\begin{prop}
Let the series $\sum_{k=1}^{\infty}a_{km}$ converge uniformly in
$m=0,1,\ldots,$ and $a_{km}\to_{m\to\infty}0,$ then $\sum_{k=1}^{\infty}a_{km}\to_{m\to\infty}0.$
\end{prop}
Using this proposition we have: \[
\|\pi P^{m}(t,s)-\pi P(t,s)\|_{\alpha}^{+}\leq\sum_{k=1}^{\infty}\|A_{k}-A_{k}^{m}\|_{\alpha}^{+}\to_{m\to\infty}0.\]
 Therefore $\|\pi P^{m}(t,s)-\pi P(t,s)\|_{\alpha}^{+}\to_{m\to\infty}0.$
\qed
\begin{cor}
\label{consequence1} The matrices $P(\cdot,\cdot)$ satisfy the Chapman-Kolmogorov
equations.
\end{cor}
\textit{Proof.} It is apparent that for all $t\leq u\leq s,$ and
$m\in\mathbb{N},$ the Chapman-Kolmogorov equations hold \[
P^{m}(t,u)P^{m}(u,s)=P^{m}(t,s).\]
 Fix some $\pi\in B_{+}^{\alpha},$ then \[
\begin{split} & \pi[P(t,u)P(u,s)-P(t,s)]=\pi[(P(t,u)-P^{m}(t,u))P(u,s)+\\
 & +P^{m}(t,u)(P(u,s)-P^{m}(u,s))+(P^{m}(t,s)-P(t,s))]=\\
 & \pi[(P(t,u)-P^{m}(t,u))P(u,s)+(P^{m}(t,u)-P(t,u))(P(u,s)-P^{m}(u,s))+\\
 & +P(t,u)(P(u,s)-P^{m}(u,s))+(P^{m}(t,s)-P(t,s))].\end{split}
\]
 Using Lemma \ref{lemma4} and the uniform boundness in the norm $P^{m}(t,s)$
on the segment $[t,s]$ (easy to check), we obtain, in the limit $m\to\infty,$
the required statement. \qed
\begin{cor}
\label{consequence2} The matices $P(\cdot,\cdot)$ are stochastic.
\end{cor}
\textit{Proof.} Let $h\in B_{\alpha}^{-}$ be the vector which consists
of all 1's (i.e. for any $i\in\mathbb{Z},h_{i}=1$), then for any
$m\in\mathbb{N},$ $t\leq s$ we have: \[
(P^{m}(t,s)h)_{i}=\sum_{j\in\mathbb{Z}}(P^{m}(t,s))_{ij}=1\Longrightarrow P^{m}(t,s)h=h.\]
 Using Lemma \ref{lemma4}, we obtain in the norm \[
P^{m}(t,s)h\to P(t,s)h.\]
 Since $P^{m}(t,s)h=h,$ the latter formula implies: \[
P(t,s)h=h,\]
 but this means that: \[
\sum_{j\in\mathbb{Z}}(P(t,s))_{ij}=1.\]
\qed
\begin{rem*}
In the Corollary \ref{consequence2} we used the fact that the chain
is not exploding. In those cases when the trajectory runs to infinity
it is impossible to adjust the norm $\|\cdot\|_{\alpha}^{-}$ such
that the vector $h,$ which consists of all 1's belongs to this space.
However, and in these cases the matrix $P(\cdot,\cdot)$ can be defined,
but it will not be stochastic.\end{rem*}
\begin{cor}
\label{consequence3} The family $P(t,s),$ $t\leq s$ is continuous
in $t$ and $s$ in $\mathcal{L}(B_{\alpha}^{\pm}).$
\end{cor}
\textit{Proof.} Since $P(t,s)$ satisfies Kolmogorov-Chapman equations,
it is enough to prove that $P(t,t+t_{1})$ is continuous at zero as
a function of $t_{1}.$ Using the formula (\ref{trivial_series})
and the estimates analogous to those of Lemma \ref{lemma3}, we have:
\[
\begin{split} & \|P(t,t+t_{1})-\mathrm{Id}\|_{\alpha}^{\pm}\leq\|e^{\int_{t}^{t+t_{1}}H_{0}(s)ds}-\mathrm{Id}\|_{\alpha}^{\pm}+\\
 & +(\exp(\mathrm{const\cdot}t_{1}\sup_{u\in[t,t+t_{1}]}(e^{-M(u)}+e^{L(u)}))-1)\to_{t_{1}\to0}0.\end{split}
\]
 This implies our statement. \qed
\begin{lem}
\label{lemma5} Distribution $p(t)\in B_{\alpha}^{+}$ as a function
of time is real analytic on $\mathbb{R}_{+}.$ The solution of (\ref{mainp})
is unique in the class of real analytic functions on $\mathbb{R}_{+}$.

Moreover, the solution of (\ref{mainp}) is unique in the class of
continuous functions $p(t)$ in $B_{\alpha}^{+}$ .\end{lem}
\begin{proof}
It is easy to check that for the remainder $R_{n}(t)$ of the series
($\ref{trivial_{s}eries}$) we have $\|R_{n}(t)\|_{\alpha}^{+}=O(t^{n}).$
Moreover, it is easy to check that $p_{k}(t)$ is infinitely differentiable.
These two statements imply the result.

Let $p^{1}(t)$ and $p^{2}(t)$ are two different analytic solutions
of (\ref{mainp}), then their difference $p(t)=p^{1}(t)-p^{2}(t)$
is a solution of (\ref{mainp}) with trivial initial data \[
p(0)\equiv0.\]
 Then it is easy to check using induction in $l$, that for any $k\in\mathbb{Z},l\in\mathbb{Z}_{+}$
\[
p_{k}(0)^{(l)}=0.\]
 Then due to the condition of the lemma $p_{k}(t)=0$ for any $k\in\mathbb{Z}.$

The proof of the last assertion is similar to the calculations made
in Lemma $\ref{lem:usefulFormula}$. Actually since

\[
\frac{d}{dz}\left(p(z)e^{\int_{z}^{t}H_{0}(s)ds}\right)=p(z)V(z)e^{\int_{z}^{t}H_{0}(s)ds},\]

integrating from 0 to $t$ we get \[
p(t)-p(0)e^{\int_{0}^{t}H_{0}(s)ds}=\int_{0}^{t}p(z)V(z)e^{\int_{z}^{t}H_{0}(s)ds}dz.\]

Now iterating this formula we conclude:

\[
p(t)=p(0)e^{\int_{0}^{t}H_{0}(s)ds}+\sum_{k=1}^{n}\int_{\triangle_{k}(0,t)}p(0)e^{\int_{0}^{s_{k}}H_{0}(s)ds}V(s_{k})\ldots V(s_{1})e^{\int_{s_{1}}^{t}H_{0}(s)ds}ds_{k}\ldots ds_{1}+R_{n+1},\]

where

\[
R_{n+1}=\int_{\triangle_{n+1}(0,t)}p(s_{n+1})V(s_{n+1})\ldots V(s_{1})e^{\int_{s_{1}}^{t}H_{0}(s)ds}ds_{n+1}\ldots ds_{1}\]

converges to 0.
\end{proof}
Thus, we constructed family $P(\cdot,\cdot),$ which consists of stochastic
matrices and satisfies (\ref{mainp}).

\paragraph{Local existence and uniqueness of the nonlinear system.}

In this section we consider the original problem with $Z=(L,M)$ which
satisfies the system of differential equations ($\ref{mainp}$)-($\ref{mainM}$).
From the formal point of view, to prove that the process is defined
on the segment $[0,T]$ ($T\geq0$ is arbitrary), it is necessary
to solve the infinite system of differential equations for the pair
$(p,Z).$ However, as it was shown above for any $Z\in C[0,T]\times C[0,T],$
there exists the unique Markov process $X_{Z},$ having transition
probabilities $P_{Z}(\cdot,\cdot),$ the infinitesimal matrix $H_{Z}(t)$
and the distribution $p_{Z}(t)$ at the moment $t$. By substitution
one can get a closed system of differential equations for $Z$: \begin{equation}
\left\{ \begin{array}{ll}
L'(t)= & -p(0)P_{Z}(0,t)\lambda_{Z}(t)+C_{\lambda},\\
M'(t)= & +p(0)P_{Z}(0,t)\mu_{Z}(t)-C_{\mu}\end{array}\right.\label{trivial6'}\end{equation}
 with the initial data \[
\left\{ \begin{array}{ll}
L(0)=L_{0},\\
M(0)=M_{0},\end{array}\right.\]
 where $\lambda_{Z}(t)$ and $\mu_{Z}(t)$ are transition rates for
$X_{Z}(t).$ Thus one takes out $p(t)$ from consideration.

Introduce the necessary notation. Fix some $R>\max(|L_{0}|,|M_{0}|).$
Let $B(T,R)=\{f\in C[0,T]:\max_{t\in[0,T]}|f(t)|\leq R\}$ be the
closed ball in the space $C[0,T]$ of continuous functions on $[0,T]$,
equipped with the uniform metrics $\rho_{B(T,R)}$. We consider the
space $B(T,R)^{2}=B(T,R)\times B(T,R)$ with the metrics \[
\rho(Z_{1},Z_{2})=\rho((L_{1},M_{1}),(L_{2},M_{2}))=\rho_{B(T,R)}(L_{1},L_{2})+\rho_{B(T,R)}(M_{1},M_{2}).\]
 It will be convenient to take the parameters $L_{0},M_{0}$ equal
to zero. This can be done by shifting the coordinates.

In the estimates below we will use some unknown functions of initial
data $L_{0},M_{0},R,T$ etc. By $\mathrm{c}(\ldots)$ we denote any
nonnegative function, nondecreasing in each of its arguments.
\begin{lem}
\label{lemma6} Let $p(0)\in B_{\alpha}^{+},$ $Z_{1},Z_{2}\in B(T,R)^{2},$
then for any $t\in[0,T]$ \[
\|p(0)P_{Z_{1}}(0,t)-p(0)P_{Z_{2}}(0,t)\|_{\alpha-1}^{+}\leq\mathrm{c}(R,T,\|p(0)\|_{\alpha}^{+},|L_{0}|,|M_{0}|)\rho(Z_{1},Z_{2}).\]
 (The left side of the inequality is defined, since $B_{\alpha}^{+}\subset B_{\alpha-1}^{+}.$)
\end{lem}
\textit{Proof.} First we prove that $H_{Z_{1}}(t)-H_{Z_{2}}(t)$ is
a family of bounded, continuous in $t$ operators, acting from $B_{\alpha}^{+}$
into $B_{\alpha-1}^{+}$. Let us estimate their norm. For arbitrary
$\nu\in B_{\alpha}^{+}$ we obtain: \[
\begin{split} & \|\nu H_{Z_{1}}(t)-\nu H_{Z_{2}}(t)\|_{\alpha-1}^{+}\leq\|\nu V_{Z_{1}}(t)-\nu V_{Z_{2}}(t)\|_{\alpha-1}^{+}+\\
 & +\|\nu H_{0Z_{1}}(t)-\nu H_{0Z_{2}}(t)\|_{\alpha-1}^{+}:=I_{1}(t)+I_{2}(t).\end{split}
\]
 Let us estimate each term separately.

\textbf{a.} We have \[
\begin{split} & I_{1}(t)=\|\nu V_{Z_{1}}(t)-\nu V_{Z_{2}}(t)\|_{\alpha-1}^{+}\leq\\
 & \leq\sum_{k\in\mathbb{Z}}|\beta(k-1)e^{-k+1+L_{1}(t)+L_{0}}\nu_{k-1}+\beta(k)e^{k+1-M_{1}(t)-M_{0}}\nu_{k+1}-\\
 & -\beta(k-1)e^{-k+1+L_{2}(t)+L_{0}}\nu_{k-1}-\beta(k)e^{k+1-M_{2}(t)-M_{0}}\nu_{k+1}|e^{\frac{k^{2}}{2}+(\alpha-1)|k|}\leq\\
 & \mathrm{\leq const}\sum_{k\in\mathbb{Z}}|e^{L_{1}(t)+L_{0}}-e^{L_{2}(t)+L_{0}}|\cdot|e^{-k+1}\nu_{k-1}e^{\frac{k^{2}}{2}+(\alpha-1)|k|}|+\\
 & +\mathrm{const}\sum_{k\in\mathbb{Z}}|e^{-M_{1}(t)-M_{0}}-e^{-M_{2}(t)-M_{0}}|\cdot|e^{k+1}\nu_{k+1}e^{\frac{k^{2}}{2}+(\alpha-1)|k|}|\leq\\
 & \leq\mathrm{const}\|\nu\|_{\alpha-1}^{+}|e^{L_{1}(t)+L_{0}}-e^{L_{2}(t)+L_{0}}|+\mathrm{const}\|\nu\|_{\alpha-1}^{+}|e^{-M_{1}(t)-M_{0}}-e^{-M_{2}(t)-M_{0}}|\leq\\
 & \leq\mathrm{c}(R,|L_{0}|,|M_{0}|)\rho(Z_{1},Z_{2})\|\nu\|_{\alpha-1}^{+}\leq\mathrm{c}(R,|L_{0}|,|M_{0}|)\rho(Z_{1},Z_{2})\|\nu\|_{\alpha}^{+}.\end{split}
\]

\textbf{b.} Note, that $H_{0Z_{1,2}}(t)$ are diagonal matrices, therefore
\[
\begin{split} & I_{2}(t)\leq\sum_{k\in\mathbb{Z}}|\beta(k)e^{-k+L_{1}+L_{0}}+\beta(k-1)e^{k-M_{1}-M_{0}}-\beta(k)e^{-k+L_{2}+L_{0}}-\beta(k-1)e^{k-M_{2}-M_{0}}|\cdot\\
 & \cdot|\nu_{k}|\cdot e^{\frac{k^{2}}{2}+(\alpha-1)|k|}\leq\mathrm{const}\sum_{k\in\mathbb{Z}}|e^{L_{1}+L_{0}}-e^{L_{2}+L_{0}}|\cdot|\nu_{k}|\cdot e^{-k+\frac{k^{2}}{2}+(\alpha-1)|k|}+\\
 & +\mathrm{const}\sum_{k\in\mathbb{Z}}|e^{-M_{1}-M_{0}}-e^{-M_{2}-M_{0}}|\cdot|\nu_{k}|\cdot e^{k+\frac{k^{2}}{2}+(\alpha-1)|k|}\leq\\
 & \leq\mathrm{c}(R,|L_{0}|,|M_{0}|)\rho(Z_{1},Z_{2})\|\nu\|_{\alpha}^{+}.\end{split}
\]
 Using \textbf{a.} and \textbf{b.}, we have: \begin{equation}
\|H_{Z_{1}}(t)-H_{Z_{2}}(t)\|\leq\mathrm{c}(R,|L_{0}|,|M_{0}|)\rho(Z_{1},Z_{2}),\label{eq:estimateForHDifference}\end{equation}
where $\|\cdot\|$ is a standart supremum norm. Continuity in $t$
of $H_{Z_{1}}(t)-H_{Z_{2}}(t)$ can be checked in the same way.

Now we are going to prove a useful inequality which we will employ
later. Consider $\Gamma(t)=P_{Z_{1}}(0,t)P_{Z_{2}}(t,T),$ $t\in[0,T]$
and differentiate \[
\frac{d}{dt}\Gamma(t)=P_{Z_{1}}(0,t)(H_{Z_{1}}(t)-H_{Z_{2}}(t))P_{Z_{2}}(t,T).\]
 Due to the results of section 1 the following sequence of transformations
holds \[
B_{\alpha}^{+}\longrightarrow_{P_{Z_{1}}(0,t)}B_{\alpha}^{+}\longrightarrow_{H_{Z_{1}}(t)-H_{Z_{2}}(t)}B_{\alpha-1}^{+}\longrightarrow_{P_{Z_{2}}(t,T)}B_{\alpha-1}^{+}.\]
 Therefore $\frac{d}{dt}\Gamma(t)$ is a family of bounded and continuous
in $t$ operators acting from $B_{\alpha}^{+}$ into $B_{\alpha-1}^{+}.$
Integrating from 0 to $T,$ we obtain: \[
P_{Z_{1}}(0,T)-P_{Z_{2}}(t,T)=\int_{0}^{T}P_{Z_{1}}(0,t)(H_{Z_{1}}(t)-H_{Z_{2}}(t))P_{Z_{2}}(t,T)dt.\]
 Then

\[
\begin{split} & I:=\|p(0)P_{Z_{1}}(0,t)-p(0)P_{Z_{2}}(0,t)\|_{\alpha-1}^{+}\leq\\
 & \leq\|p(0)\|_{\alpha}^{+}\int_{0}^{T}\|P_{Z_{1}}(0,t)\|_{\alpha}^{+}\cdot\|H_{Z_{1}}(t)-H_{Z_{2}}(t)\|\cdot\|P_{Z_{2}}(t,T)\|_{\alpha-1}^{+}dt.\end{split}
\]

Using ($\ref{eq:estimateForHDifference}$) and the estimate of Lemma
\ref{lemma3}, we get \[
I\leq\mathrm{c}(R,T,\|p(0)\|_{\alpha}^{+},|L_{0}|,|M_{0}|)\rho(Z_{1},Z_{2}).\]
 \qed
\begin{lem}
\label{lemma7} For any initial data $L_{0},M_{0}\in\mathbb{R},$
$p(0)\in B_{\alpha}^{+}$ the system (\ref{trivial6'}) has a unique
solution for $t$ sufficiently small.
\end{lem}
\textit{Proof.} First we verify that for any $Z\in B(T,R)^{2}$ the
function \[
\langle p(0)P_{Z}(0,t),\lambda_{Z}(t)\rangle\colon\mathbb{R}\to\mathbb{R}\]
 is a function continuous on the segment $[0,T]$. From Lemma \ref{consequence3}
it follows that $p(0)P_{Z}(0,\cdot)\colon\mathbb{R}\to B_{\alpha}^{+}$
is a continuous function. Moreover for any $t\in[0,T]$ and $k\in\mathbb{Z}$
\[
|(p(0)P_{Z}(0,t))_{k}|\leq\|p(0)\|_{\alpha}^{+}\|P_{Z}(0,t)\|_{\alpha}^{+}e^{-\frac{k^{2}}{2}-\alpha|k|},\]
 and \[
|(p(0)P_{Z}(0,t))_{k}\lambda_{kZ}(t)|\leq\beta(k)e^{L(t)+L_{0}}\|p(0)\|_{\alpha}^{+}\|P_{Z}(0,t)\|_{\alpha}^{+}e^{-\frac{k^{2}}{2}-(\alpha-1)|k|}.\]
 Therefore the series \[
\langle p(0)P_{Z}(0,t),\lambda_{Z}(t)\rangle=\sum_{k\in\mathbb{Z}}(p(0)P_{Z}(0,t))_{k}\lambda_{Z}(t)_{k}\]
 is majorized by a uniformly convergent series on $[0,T]$. That implies
our statement. It is proved analogously that for any $Z\in B(T,R)^{2}$
$\langle p(0)P_{Z}(0,\cdot),\mu_{Z}(\cdot)\rangle$ is a function
continuous on $[0,T]$.

Now we can rewrite (\ref{trivial6'}) in the integral form: \[
\left\{ \begin{array}{ll}
L(t)=-\int_{0}^{t}[\langle p(0)P_{Z}(0,s),\lambda_{Z}(s)\rangle-C_{\lambda}]ds,\\
M(t)=\int_{0}^{t}[\langle p(0)P_{Z}(0,s),\mu_{Z}(s)\rangle-C_{\mu}]ds.\end{array}\right.\]
 It is sufficient to check that for sufficiently small $T$ the following
map is contracting in $B(T,R)^{2}:$ \[
F:\begin{pmatrix}{}L(t)\\
M(t)\end{pmatrix}\to\begin{pmatrix}{}-\int_{0}^{t}[\langle p(0)P_{Z}(0,s),\lambda_{Z}(s)\rangle-C_{\lambda}]ds\\
\int_{0}^{t}[\langle p(0)P_{Z}(0,s),\mu_{Z}(s)\rangle-C_{\mu}]ds\end{pmatrix}.\]
 We have to check first that we can find such $T$ that the above
map maps $B(T,R)^{2}$ into itself (we assume that $R$ is fixed ).

Denote $F_{L},F_{M}$ the projections of $F$ onto the first and second
coordinates and estimate\[
|F_{L}(Z,t)|\leq\int_{0}^{t}|\langle p(0)P_{Z}(0,s),\lambda_{Z}(s)\rangle|ds+tC_{\lambda}\leq\int_{0}^{T}\|p(0)P_{Z}(0,s)\|_{\alpha}^{+}\|\lambda_{Z}(s)\|_{\alpha}^{-}ds+\]
\[
+TC_{\lambda}\leq\mathrm{c}(\|p(0)\|_{\alpha}^{+},|L_{0}|,R)\int_{0}^{T}\|P_{Z}(0,s)\|_{\alpha}^{+}ds+TC_{\lambda}\leq\]
\[
\leq T\mathrm{c}(\|p(0)\|_{\alpha}^{+},|L_{0}|,R,T)+TC_{\lambda}.\]
 From this we see that we can find $T$ such that for any $t\in[0,T]$
$|F_{L}(Z,t)|\leq R.$ Analogous estimate can be obtained for $F_{M}(Z,t),$
therefore for sufficiently small $T$ we obtain, that $F_{L}(Z,t)$
and $F_{M}(Z,t)$ can not leave $[-R,R],$ but this means exactly
that $F$ maps $B(T,R)^{2}$ into itself.

Consider the difference \[
\begin{split}F_{L}(Z_{1},t)-F_{L}(Z_{2},t)=\int_{0}^{t}[\langle p(0)P_{Z_{2}}(0,s),\lambda_{Z_{2}}(s)\rangle-\langle p(0)P_{Z_{1}}(0,s),\lambda_{Z_{1}}(s)\rangle]ds.\end{split}
\]
 Note that $|F_{L}(Z_{1},t)-F_{L}(Z_{2},t)|\leq\int_{0}^{t}(I(s)+J(s))ds,$
where \[
I(t)=|\langle p(0)P_{Z_{2}}(0,t),\lambda_{Z_{2}}(t)-\lambda_{Z_{1}}(t)\rangle|,\]
 \[
J(t)=|\langle p(0)P_{Z_{2}}(0,t)-p(0)P_{Z_{1}}(0,t),\lambda_{Z_{1}}(t)\rangle|.\]
 First we estimate $I(t)$. Note that \[
\lambda_{Z_{2}}(t)-\lambda_{Z_{1}}(t)=(e^{L_{2}(t)}-e^{L_{1}(t)})\xi,\]
 where $\xi\in B_{\alpha}^{-}$ is the vector with the components
$\xi_{n}=\beta(n)e^{-n+L_{0}}.$ Then \[
\begin{split} & I(t)\leq e^{L_{2}(t)-L_{1}(t)}\|p(0)P_{Z_{2}}(0,t)\|_{\alpha}^{+}\|\xi\|_{\alpha}^{-}\leq\mathrm{c}(T,R,|L_{0}|,|M_{0}|,\|p(0)\|_{\alpha}^{+})\rho(Z_{1},Z_{2}).\end{split}
\]
 Now let us estimate $J(t)$. From Lemma \ref{lemma6} we conclude
that \[
\begin{split} & J(t)\leq\|p(0)P_{Z_{2}}(0,t)-p(0)P_{Z_{1}}(0,t)\|_{\alpha-1}^{+}\|\lambda_{Z_{1}}(t)\|_{\alpha-1}^{-}\leq\\
 & \leq\mathrm{c}(|L_{0}|,R)\|p(0)P_{Z_{2}}(0,t)-p(0)P_{Z_{1}}(0,t)\|_{\alpha-1}^{+}\leq\\
 & \leq\mathrm{c}(R,T,\|p(0)\|_{\alpha}^{+},|L_{0}|,|M_{0}|)\rho(Z_{1},Z_{2}).\end{split}
\]
 Thus, we have \[
\begin{split}|F_{L}(Z_{1},t)-F_{L}(Z_{2},t)|\leq\mathrm{c}(R,T,\|p(0)\|_{\alpha}^{+},|L_{0}|,|M_{0}|)t\rho(Z_{1},Z_{2}).\end{split}
\]
 Analogous expression we get for $|F_{M}(Z_{1},t)-F_{M}(Z_{2},t)|.$
Then for sufficiently small $T$ (remind that $\mathrm{c}(\ldots)$
is nondecreasing function in each argument) we obtain that the map
is contracting. \qed

Using Lemma \ref{lemma7} and the contraction, we obtain existence
and uniqueness of the solution of the system (\ref{trivial6'}) for
small $t$. It remains to prove that this solution can be extended
to the entire axes.

\paragraph{Global existence.}

In the previous section we have proved the local existence and
uniqueness of our process. In this section we will prove that the
process can be extended to all $\mathbb{R}_{+}.$ For this purpose
it is sufficient to prove that $L(t)$ and $M(t)$ can not run off
to the infinity. In other words there is no explosion in our
model.
\begin{lem}
There exist positive nondecreasing functions $f_{1}$ and $f_{2}$
defined on $\mathbb{R}_{+},$ such that \[
-f_{1}(t)\leq L(t)\leq L_{0}+C_{\lambda}\cdot t;\]
\[
M_{0}-C_{\mu}\cdot t\leq M(t)\leq f_{2}(t).\]
\end{lem}
\begin{proof}
Note that integrating ($\ref{trivial6'}$) we get the first assertion:
\[
L(t)\leq L_{0}+C_{\lambda}\cdot t;\quad M(t)\geq M_{0}-C_{\mu}\cdot t;\]

Therefore using ($\ref{eq:estimateForV}$)

\[
\|P_{Z}(0,t)\|_{\alpha}^{+}\leq\mathrm{const}\cdot\exp(te^{C_{\mu}\cdot t-M_{0}}+te^{C_{\lambda}\cdot t-L_{0}}):=f(t).\]

We see that $f(t)$ is a positive nondecreasing function on $\mathbb{R}_{+}$.
Applying ($\ref{trivial6'}$) we have\[
\begin{split} & L'(t)\geq-\|p(0)P_{Z}(0,t)\|_{\alpha}^{+}\|\xi\|_{\alpha}^{-}e^{L(t)}+C_{\lambda}\end{split}
\geq-Cf(t)e^{L(t)}+C_{\lambda},\]
 where $C>0$ is an arbitrary sufficiently large constant, and $\xi=\left\{ e^{-n}\right\} $.
This inequality implies that if $L(\cdot)\leq\ln\frac{C_{\lambda}}{Cf(\cdot)},$
thus \[
L'(\cdot)\geq0.\]

Taking $C$ large enough and such that $L(0)\geq\ln\frac{C_{\lambda}}{Cf(0)}$
we conclude \[
L(t)\geq\ln\frac{C_{\lambda}}{Cf(t)}:=f_{1}(t).\]
 Analogous estimates can be made for $M(t).$
\end{proof}
The Theorem 1 is proved.

\subsection{Stationary points and the conserved integral}

Here we will prove Theorem 2. We already saw that if at least one
fixed point exists then $C_{\lambda}=C_{\mu}$. Now we will find
the fixed points explicitly.

Let $(\pi,L,M)$ is a fixed point of $X$. Note that $L,M$ are real
numbers. The invariant measure $\pi$ can be uniquely identified with
our discrete gaussian measure, introduced above \[
\pi_{n}=\frac{1}{\Xi}e^{-(n-s)^{2}},\]
 where $\Xi=\sum_{n\in\mathbb{Z}}e^{-(n-s)^{2}}.$

Thus we are left with the following two equations \[
\left\{ \begin{array}{ll}
\sum\limits _{k\in\mathbb{Z}}\frac{1}{\Xi}e^{-(k-s)^{2}}\beta(k)e^{-k+L}-C_{\lambda}=0;\\
\sum\limits _{k\in\mathbb{Z}}\frac{1}{\Xi}e^{-(k-s)^{2}}\beta(k-1)e^{k-M}-C_{\mu}=0.\end{array}\right.\]
Let us rewrite the first equation in terms of $s$ and $d$ \[
C_{\lambda}=\sum\limits _{k\in\mathbb{Z}}\frac{1}{\Xi}e^{-(k-s)^{2}}\beta(k)e^{-k+L}=\sum\limits _{k\in\mathbb{Z}}\frac{e^{-(k-s)^{2}}}{\Xi}\beta(k)e^{-k+s+d}=e^{d}\frac{\sum_{k}\beta(k)e^{-(k-s)^{2}}e^{-k+s}}{\sum_{l}e^{-(l-s)^{2}}}.\]
 The expression for $d$ follows \[
d=\ln\left[C_{\lambda}\left(\frac{\sum_{l}e^{-(l-s)^{2}}}{\sum_{k}\beta(k)e^{-(k-s)^{2}}e^{-k+s}}\right)\right].\]
 Then \[
L=s+d=s+\ln\left[C_{\lambda}\left(\frac{\sum_{l}e^{-(l-s)^{2}}}{\sum_{k}\beta(k)e^{-(k-s)^{2}}e^{-k+s}}\right)\right]\]
and similarly

\[
M=s-\ln\left[C_{\lambda}\left(\frac{\sum_{l}e^{-(l-s)^{2}}}{\sum_{k}\beta(k)e^{-(k-s)^{2}}e^{-k+s}}\right)\right].\]

Now we will prove Theorem 3. Namely, we will show that the system
of equations (\ref{mainp},\ref{mainL},\ref{mainM}) has the following
integral of motion \begin{equation}
K(t)=2s(t)+\sum_{k\in\mathbb{Z}}kp_{k}(t)=L(t)+M(t)+\sum_{k\in\mathbb{Z}}kp_{k}(t).\label{invariant}\end{equation}

Summing up equations for $L$ and $M,$ we have \[
(L+M)'(t)=-\sum_{k\in\mathbb{Z}}p_{k}(t)\lambda_{k}(t)+\sum_{k\in\mathbb{Z}}p_{k}(t)\mu_{k}(t)=-(\sum_{k\in\mathbb{Z}}kp_{k}(t))'.\]
It remains to prove that $(\sum_{k\in\mathbb{Z}}kp_{k}(t))'=\sum\limits _{k\in\mathbb{Z}}p_{k}(t)(\lambda_{k}(t)-\mu_{k}(t)).$
Using estimates (\ref{derivative}) (we will prove it below), it is
easy to see that the series $\sum_{k\in\mathbb{Z}}kp_{k}(t)$ can
be differentiated term by term. Thus: \[
\begin{split} & (\sum_{k\in\mathbb{Z}}kp_{k}(t))'=\sum_{k\in\mathbb{Z}}k(p_{k-1}\lambda_{k-1}-(\lambda_{k}p_{k}+\mu_{k}p_{k})+\mu_{k+1}p_{k+1})=\\
 & \sum_{k\in\mathbb{Z}}\{((k+1)-k)\lambda_{k}p_{k}+(-k+(k-1))\mu_{k}p_{k}\}=\sum_{k\in\mathbb{Z}}p_{k}(t)(\lambda_{k}-\mu_{k})\text{.}\end{split}
\]

\subsection{Convergence.}

Here we will prove Theorem 4 assuming both conditions on $\beta(n)$,
introduced above, namely (\ref{con1}) and (\ref{con2}).

We will define two Lyapunov functions. The first of them $Q(t)$ will
be positive and decreasing along the trajectory outside some special
set. The second \textbf{$W(t)$} decreases along a trajectory everywhere,
but can take big negative values. Using these two functions we can
prove the convergence.

\paragraph{Boundness of $L$ and $M$.}

Note that the pair of equations for $L$ and $M$ are equivalent to
the following pair of equations \begin{equation}
\left\{ \begin{array}{ll}
s'(t)=-\frac{1}{2}e^{d}\left(\sum_{n\in\mathbb{Z}}p_{n}\beta(n)e^{-n+s}-\sum_{n\in\mathbb{Z}}p_{n}\beta(n-1)e^{n-s}\right);\\
d'(t)=-\frac{1}{2}e^{d}\left(\sum_{n\in\mathbb{Z}}p_{n}\beta(n)e^{-n+s}+\sum_{n\in\mathbb{Z}}p_{n}\beta(n-1)e^{n-s}\right)+C_{\lambda}.\end{array}\right.\label{eq:for_s_and_d}\end{equation}
 We will use the following
\begin{prop}
\label{propostion1} If the terms of the absolutely convergent series
$u(x)=\sum_{n=1}^{\infty}u_{n}(x)$ are continuously differentiable
on the segment $[a,b]$ and the series of derivatives $\sum_{n=1}^{\infty}u_{n}'(x)$
converges uniformly in $(a,b),$ then \[
\frac{d}{dx}\sum_{n=1}^{\infty}u_{n}(x)=\sum_{n=1}^{\infty}u_{n}'(x).\]

\end{prop}
In order to use this statement we have to make additional estimates.
If $p(0)\in B_{\alpha}^{+},$ then $\|p(t)\|_{\alpha}^{+}\leq\mathrm{const}$
for any sufficiently small segment $[0,T]$. This implies for any
$t\in[0,T]$ \begin{equation}
|p_{k}(t)|\leq e^{-\frac{k^{2}}{2}-\alpha|k|}\|p_{k}(t)\|_{\alpha}^{+}\leq\mathrm{const}\cdot e^{-\frac{k^{2}}{2}-\alpha|k|}.\label{eq:auxililiaryEstimate}\end{equation}
 From this, using the formula for $p_{k}'(t),$ it is easy to see
that for any $t\in[0,T]$ \begin{equation}
|p_{k}'(t)|\leq\mathrm{const}e^{-\frac{k^{2}}{2}-(\alpha-1)|k|}.\label{derivative}\end{equation}

Moreover, we will need the following Grownall type result
\begin{prop}
\label{propostion2} Let $f(\cdot)$ be some differentiable function
on $[0,\infty)$ such that \[
f'(\cdot)\geq(\leq)g(\cdot)(C_{1}-C_{2}f(\cdot)),\quad C_{2}>0\]
where $g(\cdot)$ is a positive function. Then $f(\cdot)$ is bounded
from below (above).
\end{prop}
Define the first Lyapunov function (compare with $d'(t)$) \[
Q(t)=\sum_{n\in\mathbb{Z}}p_{n}(t)\left(\beta(n)e^{-n+s}+\beta(n-1)e^{n-s}\right).\]

\begin{lem}
\label{lemma9}\textbf{1.} $Q(t)$ is bounded. \textbf{2.} $|\sum_{k\in\mathbb{Z}}kp_{k}(t)-s(t)|$
is bounded.\end{lem}
\begin{proof}
\textbf{1.} Using estimate ~(\ref{derivative}), it is easy to check,
that the series for $Q'(t)$ can be differentiated term by term. First
we find $Q'(t)$\[
Q'(t)=\sum_{n\in\mathbb{Z}}p_{n}'\left(\beta(n)e^{-n+s}+\beta(n-1)e^{n-s}\right)+\sum_{n\in\mathbb{Z}}p_{n}(\beta(n)e^{-n+s}-\beta(n-1)e^{n-s})s'.\]
Using ($\ref{eq:for_{s}{}_{a}nd_{d}}$) we note that the second component
of the expression is less than $0$. So \[
Q'(t)\leq\sum_{n\in\mathbb{Z}}p_{n}'\left(\beta(n)e^{-n+s}+\beta(n-1)e^{n-s}\right)=e^{-d}\sum_{n\in\mathbb{Z}}p_{n}'(t)(\lambda_{n}(t)+\mu_{n}(t)).\]
Using Kolmogorov's equations and opening the brackets we have \[
Q'(t)\leq e^{-d}\sum_{n\in\mathbb{Z}}p_{n}(t)\left\{ \lambda_{n}\lambda_{n+1}+\mu_{n}\mu_{n-1}+\lambda_{n}\mu_{n+1}+\mu_{n}\lambda_{n-1}-2\lambda_{n}\mu_{n}-\lambda_{n}^{2}-\mu_{n}^{2}\right\} .\]
Substituting expressions ($\ref{eq:lambdaMu}$) for $\lambda_{n}$
and $\mu_{n}$ we have:

\[
\begin{split} & Q'(t)\leq e^{d}\sum_{n\in\mathbb{Z}}p_{n}(t)[\beta(n)\beta(n+1)e^{-2n-1+2s}+\beta(n-1)\beta(n-2)e^{2n-1-2s}+e\beta^{2}(n)+\\
 & +e\beta^{2}(n-1)-2\beta(n-1)\beta(n)-\beta^{2}(n)e^{-2n+2s}-\beta^{2}(n-1)e^{2n-2s}]:=e^{d}\sum_{n\in\mathbb{Z}}S_{n}(s)\pi_{n}(t).\end{split}
\]
We state that for the just defined function $S_{n}(s)$ \[
S_{n}(s)+\beta(n)e^{-n+s}+\beta(n-1)e^{n-s}\leq\mathrm{const.}\]
Indeed, putting $x:=e^{s-n},$ we can rewrite this inequality in the
following form:

\[
\beta(n)\left[\left(\frac{1}{e}\beta(n+1)-\beta(n)\right)x^{2}+x\right]+\beta(n-1)\left[\left(\frac{1}{e}\beta(n-2)-\beta(n-1)\right)x^{2}+x\right]\leq\mathrm{const.}\]
Therefore conditions (\ref{con1}) and (\ref{con2}) imply required
inequality. Now we get \[
Q'(t)\leq e^{d}\sum_{n\in\mathbb{Z}}p_{n}(t)(\mathrm{const}-\beta(n)e^{-n+s}-\beta(n-1)e^{n-s})=e^{d}(\mathrm{const}-Q(t)).\]
 Using Proposition \ref{propostion2}, we obtain that $Q$ is bounded.

\textbf{2.} We have \[
\begin{split} & |\sum_{n\in\mathbb{Z}}np_{n}(t)-s(t)|=|\sum_{n\in\mathbb{Z}}np_{n}-s|=\sum_{n\in\mathbb{Z}}|p_{n}(n-s)|\leq\\
 & \leq\mathrm{const}\sum_{n\in\mathbb{Z}}|p_{n}|\left(\beta(n)e^{-n+s}+\beta(n-1)e^{n-s}\right)=\mathrm{const\cdot}Q(t).\end{split}
\]
 It remains to use the boundness of $Q$.\end{proof}
\begin{lem}
\label{lemma10} Functions $L$ and $M$ are bounded.\end{lem}
\begin{proof}
Note that this statement is equivalent to the fact that $s$ and $d$
are bounded. Boundness of $s$ directly follows from formula (\ref{invariant})
and part 2 of Lemma \ref{lemma9}. Boundness of $d$ is proved as
Proposition \ref{propostion2}, since \[
d'(t)=-\frac{1}{2}e^{d}Q(t)+C_{\lambda}.\]

\end{proof}

\paragraph{Relative entropy for constant $L$ and $M$.}

In the nonlinear case we will prove convergence using the relative
entropy method. In this subsection we will introduce auxiliary notions
and lemmas for constant $Z.$

If $Z=\mathrm{const}$ there exists only one invariant measure
$\pi,$ given by ($\ref{eq:invariantMeasure}$). Define the entropy
of the distribution $p=p(t)$ relative to $\pi$ in the following
way

\begin{equation}
H(t)=H(p(t)|\pi)=\sum_{n\in\mathbb{Z}}p_{n}\ln\frac{p_{n}}{\pi_{n}}=\sum_{n\in\mathbb{Z}}\pi_{n}\varphi\left(\frac{p_{n}}{\pi_{n}}\right),\label{eq:entropy}\end{equation}
where $\varphi(x)=x\ln x$.
\begin{rem*}
As the factor $\Xi$ adds a constant to $H(t)$, everywhere below
we will assume $\Xi=1.$ Thus $p(t)$ is just a finite (not necessary
probability) measure.

The fact that $H(t)$ decreases in time is known {[}2{]}. We will
show that the series ($\ref{eq:entropy}$) is convergent and can be
differiantated term by term. We will use the following technical lemma: \end{rem*}
\begin{lem}
\label{lem:entropyLemma1}For any $\varepsilon>0$ there exists C=C($\varepsilon)$
> 0 such that for $t>\varepsilon$ \[
\ln p_{n}(t)\geq-Ce^{|n|}(1+t).\]

\end{lem}
\textit{Proof. }Let \textit{$n\in\mathbb{Z}_{+},$ $p_{0}(0)>0.$}
We will make a very rough estimate. It is evident that $p_{n}(t)$
is greater than the product of the following probabilities:
\begin{itemize}
\item The probability $P_{1}$ that at the moment $t=0$ the particle is
at the point 0.
\item The probability $P_{2}$ that the only jumps before time $t$ are
as follows: from 0 to 1, from 1 to 2, etc., from $n-1$ to $n.$
\end{itemize}
Otherwise speaking

\[
\begin{split} & p_{n}(t)\geq p_{0}(0)\cdot\left[\prod_{k=0}^{n-1}\left(\frac{\beta(n)e^{-n+L}}{\beta(n-1)e^{n-M}+\beta(n)e^{n-L}}\right)\right]\cdot\left[e^{-(\beta(n-1)e^{n-M}+\beta(n)e^{-n+L})t}\right]\cdot\\
 & \cdot\left[\frac{\left(\mathrm{min}_{n}\left(\beta(n-1)e^{n-M}+\beta(n)e^{-n+L}\right)t\right)^{n}}{n!}e^{-\mathrm{min}_{n}\left(\beta(n-1)e^{n-M}+\beta(n)e^{-n+L}\right)t}\right]\geq\\
 & \geq\pi_{0}(0)C^{n}e^{-n(n-1)}\cdot e^{-Ce^{n}t}\cdot\frac{\left(ct\right)^{n}}{n!}e^{-ct}.\end{split}
\]

Taking the logarithm, we get the requested assertion. Similar calculation
can be made in the case $p_{0}(0)=0,$ $p_{k}(0)\neq0,k\neq0,$ and
$n\in\mathbb{Z}^{-}.$ \qed
\begin{cor}
The series ($\ref{eq:entropy}$) is well defined for $t\geq0$, and
for $t>0$ it can be differiantiated term by term.\end{cor}
\begin{proof}
Whereas ($\ref{eq:auxililiaryEstimate}$) it is easy to see that the
series ($\ref{eq:entropy}$) converges on any sufficiently small interval.
Using Lemma $\ref{lem:entropyLemma1}$ we get that the corresponding
series of derivatives uniformly converges on any interval enough small.
The corollary implies Proposition $\ref{propostion1}$.\end{proof}
\begin{lem}
For $t>0$\[
\frac{d}{dt}H(t)\leq0.\]

The equality is attained as soon as $p(t)=\pi$ up to multiplicative
factor.
\end{lem}

\paragraph{Nonlinear case.}

Now we will consider the case when $Z=(L,M)$ satisfy equations (\ref{mainL},\ref{mainM}).
Similarly we define

\[
\pi_{n}^{s}=e^{-(s-n)^{2}}\]

and the relative entropy

\begin{equation}
H(t)=H(p(t)|\pi^{s})=\sum_{n\in\mathbb{Z}}p_{n}\ln\frac{p_{n}}{\pi_{n}^{s}}=\sum_{n\in\mathbb{Z}}p_{n}(\ln p_{n}+(s-n)^{2}).\label{eq:entropy2}\end{equation}

The measure $\pi^{s}$ is not stochastic because we prefer not to
normalize it.

First of all we have to check that the series ($\ref{eq:entropy2}$)
converges and can be differentiated term by term. Taking into account
that $L$ and $M$ are bounded we can prove the technical lemma similar
to Lemma $\ref{lem:entropyLemma1}$ :
\begin{lem}
Let $\varepsilon>0$ be a fixed number. Then there exists $C=C(\varepsilon)>0$
such that for $t>\varepsilon$

\[
\ln p_{n}(t)\geq-Ce^{|n|}(1+t).\]
\end{lem}
\begin{cor}
The series ($\ref{eq:entropy2}$) is well defined for $t\geq0$, and
for $t>0$ it can be differiantiated term by term.
\end{cor}
Define \[
W(t)=H(t)+2Ks(t)-3s^{2}(t),\]

where $K=2s+\sum_{n\in\mathbb{Z}}np_{n}(t)=\mathrm{const}$ is the
invariant, introduced above.
\begin{lem}
\label{lem:mainEntropyLema}For $t>0$

\[
\frac{d}{dt}W(t)\leq0.\]
The equality is attained as soon as $p(t)=\pi^{s}$ up to multiplicative
factor.\end{lem}
\begin{proof}
For $t>0$ the series ($\ref{eq:entropy2}$) can be differiantiated
term by term. Thus \begin{equation}
\frac{d}{dt}H(t)=\frac{\partial H}{\partial\pi}\cdot\frac{d\pi}{dt}+\frac{\partial H}{\partial s}\cdot\frac{ds}{dt}.\label{eq:entropy3}\end{equation}
 The first term of the right side represents a derivative of the relative
entropy for fixed $s$, therefore, using results of the preceding
subsection, we conclude that this term is negative. Let us calculate
the second term

\[
\frac{\partial H}{\partial s}\cdot\frac{ds}{dt}=\sum_{n\in\mathbb{Z}}2p_{n}(s-n)s'=2ss'\sum_{n\in\mathbb{Z}}p_{n}-2s'\sum_{n\in\mathbb{Z}}p_{n}n=2ss'-2s'\sum_{n\in\mathbb{Z}}np_{n}(t).\]
Using the invariant $K$ we get

\[
\frac{\partial H}{\partial s}\cdot\frac{ds}{dt}=2ss'-2s'(K-2s)=6ss'-2s'K=(3s^{2}-2Ks)'\]
and by ($\ref{eq:entropy3}$) we get \[
\frac{d}{dt}\{H+2Ks-3s^{2}\}\leq0.\]
\end{proof}
\begin{cor}
Let $p(0)\in B_{\alpha}^{+}$ and $L_{0,}M_{0},\alpha\in\mathbb{R}$
be arbitrary numbers. Then the convergence holds.\end{cor}
\begin{proof}
In Lemma $\ref{lem:mainEntropyLema}$ we introduced the function $W=W(p,Z)$,
which can be considered as the Lyapunov function. Therefore the proof
has a quite standart scheme.

First, $p(t)$ belongs to a bounded (supremum norm) closed subset
of the set $C_{0}(\mathbb{\bar{Z}})$ of functions $f$ on $\bar{\mathbb{Z}}=\mathbb{Z}\cup\{\infty\}$,
continuous at infinity and such that $f(\infty)=0.$ That is \[
\left\{ p(t)\right\} _{t\in\mathbb{R}_{+}}\subset B_{\alpha}^{+}\cap\{\mu=\{\mu_{n}\}:\mu_{n}\geq0,\sum_{n\in\mathbb{Z}}\mu_{n}=1\}\subset C_{0}(\bar{\mathbb{Z}}).\]
 It follows that the domain of $W$ lies in a compact subset of $C_{0}(\bar{\mathbb{Z}})\times\mathbb{R}^{2}$,
as $Z=(L,M)$ is bounded.

Therefore the trajectory $\{(p(t),Z(t))\}$ has at least one limiting
point. Let $(\pi^{*},Z^{*})$ be one of such points. As $W$ decreases
along the trajectory

\[
\frac{d}{dt}W_{(\pi^{*},Z^{*})}(t)=0.\]

This follows from continuous differentiability of $W(t)$. Using Lemma
$\ref{lem:mainEntropyLema}$ we conclude that $\pi_{n}^{*}=e^{-(s^{*}-n)^{2}}$
up to some factor.

Using the invariant it is easy to check that $s^{*}$ is defined
uniquely. In the introduction we mentioned that the trajectory of
$(p,s)$ does not depend on the choice of $d$. Therefore $(p,s)$
converges to $(\pi^{*},s^{*}).$ It remains to show that $d$
converges to $d^{*}.$ Remind that we can rewrite the equation for
$d$ in the following form \[
d'=-\frac{1}{2}e^{d}Q(t)+C_{\lambda},\] where \[
Q(t)=\sum_{n\in\mathbb{Z}}p_{n}\beta(n)e^{-n+s}+\sum_{n\in\mathbb{Z}}p_{n}\beta(n-1)e^{n-s}.\]
Note that due to the established convergence of $s$ and $p$ \[
Q(t)\to_{t\to\infty}Q^{*}=\mathrm{const.}\] We have an ordinary
differential equation which can be solved explicitly

\[
e^{-d(t)}=C_{1}e^{-C_{\lambda}\cdot t}+\frac{1}{2}e^{-C_{\lambda}\cdot t}\int_{0}^{t}e^{C_{\lambda}\cdot s}Q(s)ds.\]
As $Q(t)$ converges we obtain the required result.\end{proof}

\lyxaddress{K. Vaninsky \\
Department of Mathematics \\
 Michigan State University \\
 East Lansing, MI 48824 \\
 USA vaninsky@math.msu.edu}

\lyxaddress{S. Muzychka \\
 Faculty of Mathematics and Mechanics \\
 Moscow State University \\
 Vorobjevy Gory\\
 Moscow, Russia}
\end{document}